\documentclass[11pt]{article}

\usepackage{fullpage, amsmath, amsthm, amssymb, microtype, tikz, braket}
\usetikzlibrary{patterns}

\usepackage[colorlinks, citecolor=blue, linkcolor=blue, bookmarks=true]{hyperref}
\usepackage[nameinlink, capitalize]{cleveref}

\makeatletter
\crefdefaultlabelformat{\leavevmode\sw@slant#2{\upshape#1}#3}
\makeatother

\newtheorem{prop}{Proposition}
\newtheorem{theorem}{Theorem}
\newtheorem{lem}{Lemma}

\theoremstyle{definition}

\newcommand{\N}{\mathbb{N}}

\newcommand{\C}{\mathbb{C}}
\newcommand{\F}{\mathbb{F}}

\renewcommand{\epsilon}{\varepsilon}

\renewcommand{\hat}{\widehat}

\DeclareMathOperator{\Qres}{Q}
\newcommand{\Qcc}{\Qres^{*\text{\normalfont cc}}}
\newcommand{\Qoip}{\Qres^{\text{\normalfont OIP}}}

\newcommand{\suppress}[1]{}

\DeclareMathOperator*{\E}{E}

\title{Quantum Communication-Query Tradeoffs}
\author{William M. Hoza\thanks{This material is based upon work supported by the National Science Foundation Graduate Research Fellowship under Grant No. DGE-1610403.}\\Department of Computer Science,\\University of Texas at Austin\\\texttt{whoza@utexas.edu}}

\begin{document}
	\maketitle
	
	\begin{abstract}
		For any function $f: X \times Y \to Z$, we prove that
		\[
			\Qcc(f) \cdot \Qoip(f) \cdot (\log \Qoip(f) + \log |Z|) \geq \Omega(\log |X|).
		\]
		Here, $\Qcc(f)$ denotes the bounded-error communication complexity of $f$ using an entanglement-assisted two-way qubit channel, and $\Qoip(f)$ denotes the number of quantum queries needed to learn $x$ with high probability given oracle access to the function $f_x(y) \stackrel{\text{def}}{=} f(x, y)$. We show that this tradeoff is close to the best possible. We also give a generalization of this tradeoff for \emph{distributional} query complexity.
		
		As an application, we prove an optimal $\Omega(\log q)$ lower bound on the $\Qcc$ complexity of determining whether $x + y$ is a perfect square, where Alice holds $x \in \F_q$, Bob holds $y \in \F_q$, and $\F_q$ is a finite field of odd characteristic. As another application, we give a new, simpler proof that searching an ordered size-$N$ database requires $\Omega(\log N / \log \log N)$ quantum queries. (It was already known that $\Theta(\log N)$ queries are required.)
	\end{abstract}

	\section{Introduction}
	Let $f: X \times Y \to Z$ be a function, where $X, Y, Z$ are finite sets. In this paper, we study quantum algorithms for two computational problems associated with $f$.
	
	\paragraph{The communication problem.} In this problem, Alice has $x \in X$, Bob has $y \in Y$, and the goal is for Bob to output $f(x, y)$ with as little communication between Alice and Bob as possible. Three different models have been considered for exploiting the laws of quantum mechanics to assist with this task:
	\begin{itemize}
		\item We can give Alice and Bob a (two-way) \emph{qubit channel}, i.e. we allow them to exchange qubits instead of classical bits.
		\item We can allow Alice and Bob to share an entangled state of arbitrary finite dimension that does not depend on their inputs.
		\item We can give Alice and Bob both capabilities.
	\end{itemize}
	A qubit channel can be simulated with factor-$2$ overhead using an entanglement-assisted classical channel via teleportation \cite{bbc+93}. On the other hand, no simulation in the other direction is known; there might be some functions $f$ for which entanglement is more useful than a qubit channel. In this paper, we consider the most powerful model, where Alice and Bob have an entanglement-assisted qubit channel. Let $\Qcc_{\epsilon}(f)$ be the minimum number of qubits exchanged in any protocol in this model for computing $f$ with failure probability at most $\epsilon \geq 0$. As shorthand, define $\Qcc(f) = \Qcc_{1/3}(f)$.
	
	\paragraph{The oracle identification problem.} In this problem, there is an unknown value $x \in X$, and the goal is to determine $x$ using as few queries as possible to an oracle for the function $f_x(y) \stackrel{\text{def}}{=} f(x, y)$. For example, ``hidden shift'' problems like the Bernstein-Vazirani problem \cite{bv97} fit into this framework. Classically, for any $f$, at least $\frac{\log |X|}{\log |Z|}$ queries are required, for the trivial reason that each query gives at most $\log |Z|$ bits of information about $x$. But we will allow quantum queries, i.e. the algorithm can apply $U_{f_x}$, the unitary operator on $\C^{|Y|} \otimes \C^{|Z|}$ defined by
	\[
		U_{f_x} \ket{y} \ket{a} = \ket{y} \ket{a + f(x, y) \text{ mod } |Z|}.
	\]
	Let $\Qoip_{\epsilon}(f)$ be the minimum number of queries made by an algorithm with failure probability at most $\epsilon \geq 0$. As shorthand, define $\Qoip(f) = \Qoip_{1/3}(f)$. Note that we allow for the possibility that $f_x \equiv f_{x'}$ for some $x \neq x'$. In this degenerate case, $\Qoip(f) = \infty$.
	
	\subsection{Main result}
	In this paper, we show that $\Qoip(f)$ and $\Qcc(f)$ cannot \emph{both} be tiny:
	\begin{theorem} \label{thm:main}
		For any\footnote{In the degenerate case that $f$ does not depend on $x$, $\Qcc(f) = 0$ and $\Qoip(f) = \infty$; the left-hand side of \cref{eqn:main} should be interpreted as evaluating to $\infty$ in this case. Similar remarks apply to \cref{thm:zero-error,thm:distributional-tradeoff,lem:distributional-tradeoff-before-amplification}.} function $f: X \times Y \to Z$,
		\begin{equation} \label{eqn:main}
		\Qcc(f) \cdot \Qoip(f) \cdot (\log \Qoip(f) + \log |Z|) \geq \Omega(\log |X|).
		\end{equation}
	\end{theorem}
	A consequence of \cref{thm:main} is that an \emph{upper} bound on one of the two complexity measures, $\Qoip(f)$ or $\Qcc(f)$, implies a \emph{lower} bound on the other complexity measure. This technique can prove very strong communication complexity lower bounds: if $\Qoip(f) \leq O(1)$, then \cref{thm:main} implies that $\Qcc(f) \geq \Omega(\log |X|)$, which is trivially optimal for any $f$. On the other hand, this technique has limited value for proving query complexity lower bounds: even if $\Qcc(f) \leq O(1)$, \cref{thm:main} only gives the lower bound $\Qoip(f) \geq \Omega\left(\frac{\log |X|}{\log |Z| \log \log |X|}\right)$, whereas in actual fact, $\Qoip(f)$ might be much larger. For example, $\Qoip(\mathsf{EQ}) \geq \Omega(\sqrt{|X|})$; this is the familiar fact that Grover's algorithm \cite{gro96} is optimal, proven by Bennett et al. \cite{bbbv97}.
	
	\subsection{Proof overview}
	
	\cref{thm:main} generalizes the technique used by Cleve et al. \cite{cdnt13} to prove that $\Qcc(\mathsf{IP}) \geq \Omega(n)$. In our terminology, Cleve et al. used the fact that $\Qoip(\mathsf{IP}) = 1$ by the Bernstein-Vazirani algorithm \cite{bv97, cemm98}. The core of the proof is a lower bound on the communication required to transmit classical data over an entanglement-assisted two-way qubit channel. Based on Holevo's theorem \cite{hol73}, Cleve et al. proved such a bound \cite[Theorem 1]{cdnt13} in terms of mutual information, a measure of the amount of data transmitted \emph{in expectation}. We will use a similar bound (\cref{thm:ns02}) by Nayak and Salzman for transmitting data \emph{with high probability}. Cleve et al.'s bound would be sufficient for \cref{thm:main}, but Nayak and Salzman's bound allows us to derive a strengthening of \cref{thm:main} for the \emph{distributional} version of $\Qoip(f)$, i.e. the case when the unknown value $x$ is drawn from a known distribution.
	
	Given that bound, the proof of \cref{thm:main} is very simple. To transmit $x \in X$ from Alice to Bob, Bob executes the $\Qoip(f)$ algorithm, and the two of them simulate each query using the $\Qcc_{\epsilon}(f)$ protocol. If $\epsilon$ is small enough, the transmission is successful with constant probability, which implies that the total number of qubits sent from Alice to Bob, $\Qoip(f) \cdot \Qcc_{\epsilon}(f)$, must be at least $\frac{1}{2} \log |X| - O(1)$. The final bound in \cref{thm:main} comes from the standard amplification bound $\Qcc_{\epsilon}(f) \leq O(\Qcc(f) \cdot \log(1/\epsilon))$.

	\subsection{Related work}
	
	Entanglement-assisted communication complexity was introduced by Cleve and Buhrman \cite{cb97}. There is a sizable body of research on proving lower bounds for $\Qcc$ \cite{bw01, raz03, mw07, sz09, ls09a, ls09b, kw13, bblv13, cdnt13}. Two papers are especially relevant to the present work. First, as previously mentioned, our technique generalizes that of Cleve et al. \cite{cdnt13}. Second, Montanaro and Winter \cite{mw07} gave a different generalization of the argument by Cleve et al. by considering a specific family of one-query algorithms to replace the Bernstein-Vazirani algorithm.
	
	Oracle identification problems were first studied in generality by Ambainis et al. \cite{aik+04}. One line of research on oracle identification problems \cite{aik+04, aik+07, ain+09, kot14} has focused on the \emph{worst-case $f$}, i.e. the quantity
	\[
			\mathrm{OIP}(X, Y, Z) \stackrel{\text{def}}{=} \max_{f: X \times Y \to Z}\Qoip(f),
	\]
	where the maximum ranges over functions $f$ such that $\Qoip(f) < \infty$. Kothari proved optimal bounds on $\mathrm{OIP}(X, Y, \{0, 1\})$ \cite{kot14}. In contrast, our tradeoff theorem provides a technique for proving lower bounds on $\Qoip(f)$ for \emph{specific} functions $f$.
	
	In another line of research \cite{sg04, as05, hmp+10}, oracle identification problems have been studied under the alternate name \emph{exact learning from membership queries}. Servedio and Gortler \cite{sg04} showed that for any $f: X \times Y \to \{0, 1\}$,
	\begin{equation} \label{eqn:sg04}
		\Qoip(f) \geq \Omega\left(\frac{\log |X|}{\log |Y|} + \sqrt{1/\gamma(f)}\right),
	\end{equation}
	where $\gamma(f)$ is a certain combinatorial parameter satisfying $2 \leq 1/\gamma(f) \leq |Y| + 1$. This bound is incomparable with ours. On the one hand, \cref{eqn:sg04} only gives a trivial $\Omega(1)$ lower bound for searching an ordered size-$N$ database, whereas we will show in \cref{sec:gt} that our tradeoff implies a near-optimal $\Omega(\log N / \log \log N)$ lower bound. On the other hand, \cref{eqn:sg04} gives the optimal $\Omega(\sqrt{N})$ lower bound for \emph{unordered} search, whereas our tradeoff only gives a weak $\Omega(\log N / \log \log N)$ lower bound.
	
	Ours is far from the first theorem relating query complexity and communication complexity. Many researchers have shown how to derive communication complexity lower bounds from \emph{lower} bounds on the query complexity of related functions \cite{rm99, gpw15, goo15, glm+16, drnv16, hhl16, goo16, wyy17, cklm17, gkpw17, gpw17, agj+17}. But to the best of our knowledge, we are the first to explicitly observe a \emph{tradeoff} between communication complexity and query complexity.
	
	We defer discussion of previous work related to our applications to \cref{sec:applications}.

	\subsection{Outline of this paper}
	In \cref{sec:zero-error}, we prove a version of \cref{thm:main} for \emph{exact} communication protocols and query algorithms, i.e. for the case $\epsilon = 0$. Then, in \cref{sec:nonzero-error}, we prove \cref{thm:main} and a generalization of \cref{thm:main} for the distributional version of $\Qoip$. In \cref{sec:applications}, we discuss applications of our tradeoff theorems. Finally, in \cref{sec:optimality}, we show that \cref{thm:main} and our tradeoff theorem for the exact case are near-optimal, thereby demonstrating the limitations of our lower bound techniques.

	\section{Tradeoff theorem for the exact case} \label{sec:zero-error}
	We begin with the exact case, where the bound is improved and the proof is very simple.
	\begin{theorem} \label{thm:zero-error}
		For any $f: X \times Y \to Z$,
		\[
			\Qcc_0(f) \cdot \Qoip_0(f) \geq \frac{1}{2} \log |X|.
		\]
	\end{theorem}
	To prove \cref{thm:zero-error}, we closely follow the analysis by Cleve et al. \cite{cdnt13}. Recall that a $\Qcc$ protocol operates on four registers: Alice and Bob's input registers; a shared register initially in a bipartite state $\ket{\phi}$ that does not depend on the inputs and that also serves as the players' ancilla registers; and an output register of dimension $|Z|$, initially in the state $\ket{0}$, held by Bob. Suppose $U$ is a unitary operator on $\C^{|X|} \otimes \C^{|Y|} \otimes \C^d \otimes \C^{|Z|}$. We say that $U$ \emph{cleanly computes $f$} if there is some state $\ket{\phi} \in \C^d$ such that for every $x \in X, y \in Y, a \in Z$,
	\[
		U \ket{x} \ket{y} \ket{\phi} \ket{a} = \ket{x} \ket{y} \ket{\phi} \ket{a + f(x, y) \text{ mod } |Z|}.
	\]
	Identify each $\Qcc$ protocol with the unitary operator\footnote{This is well-defined, because the $\Qcc$ model does not allow intermediate measurements, i.e. a $\Qcc$ protocol only involves applying unitary operators and exchanging qubits.} implemented by the protocol. The standard ``copy and uncompute'' trick proves the following lemma:
	
	\begin{lem} \label{lem:clean}
		For any function $f: X \times Y \to Z$, there is a communication protocol $U$ in which Alice sends Bob at most $\Qcc_0(f)$ qubits such that $U$ cleanly computes $f$.
	\end{lem}

	\begin{proof}
		Let $U_0$ be a protocol for $f$ in which at most $\Qcc_0(f)$ qubits are exchanged. Let $V$ be the unitary operator on $\C^{|Z|} \otimes \C^{|Z|}$ defined by
		\[
			V: \ket{z} \ket{a} \mapsto \ket{z} \ket{a + z \text{ mod } |Z|}.
		\]
		In $U$, Bob has an extra $|Z|$-dimensional ancilla register as compared to $U_0$. The protocol $U$:
		\begin{enumerate}
			\item Execute $U_0$, with the new ancilla register taking the place of the output register of $U_0$.
			\item Apply $V$ to copy the result into the actual output register.
			\item Execute $U_0^{-1}$ to restore the registers to their initial state (other than the output register).
		\end{enumerate}
		The number of qubits sent from Alice to Bob in this protocol is the number of qubits sent from Alice to Bob in $U_0$, plus the number of qubits sent from Bob to Alice in $U_0$. This is at most $\Qcc_0(f)$.
	\end{proof}

	\begin{proof}[Proof of \cref{thm:zero-error}]
		We give a protocol for sending an arbitrary element $x \in X$ from Alice to Bob. Bob executes the $\Qoip_0(f)$ protocol, simulating each query by executing the clean communication protocol for $f$. The total number of qubits sent from Alice to Bob in this protocol is only $\Qcc_0(f) \cdot \Qoip_0(f)$. But by \cite[Theorem 1]{cdnt13} (a consequence of Holevo's theorem \cite{hol73}), the number of qubits sent from Alice to Bob must be at least $\frac{1}{2} \log |X|$.
	\end{proof}
	
	\section{Tradeoff theorems for the bounded-error case} \label{sec:nonzero-error}
	
	\subsection{Smooth max-entropy}
	We will formulate our distributional communication-query tradeoff in terms of \emph{smooth max-entropy}. If $x$ and $y$ are random variables on the same measurable space, we write $x \sim_{\epsilon} y$ if $x$ and $y$ are $\epsilon$-close in total variation distance. Recall that if $x$ is a random variable and $\epsilon > 0$, the \emph{$\epsilon$-smooth max-entropy} of $x$, denoted $H_{\text{max}}^{\epsilon}(x)$, is defined by
	\begin{equation}
		H_{\text{max}}^{\epsilon}(x) = \min\{\log_2 |\text{supp}(y)| : x \sim_{\epsilon} y\}.
	\end{equation}
	Equivalently, $2^{H_{\text{max}}^{\epsilon}(x)}$ is the size of the smallest set $S$ such that $\Pr[x \in S] \geq 1 - \epsilon$. Operationally, $\lceil H_{\text{max}}^{\epsilon}(x) \rceil$ is the \emph{worst-case} number of bits needed to transmit $x$ over a classical channel with $\epsilon$ probability of failure.
	
	\subsection{Distributional communication-query tradeoff}
	For a function $f: X \times Y \to Z$, a distribution $\mu$ on $X$, and $\epsilon > 0$, let $\Qoip_{\mu, \epsilon}(f)$ be the minimum number of queries of any quantum algorithm that outputs $x$ with failure probability at most $\epsilon$ given oracle access to $f_x(y) \stackrel{\text{def}}{=} f(x, y)$, where the randomness includes both the internal randomness of the algorithm and a random draw of $x$ from distribution $\mu$. Clearly, for every distribution $\mu$, $\Qoip_{\mu, \epsilon}(f) \leq \Qoip_{\epsilon}(f)$. The distributional strengthening of \cref{thm:main}:
	\begin{theorem} \label{thm:distributional-tradeoff}
		For every $f: X \times Y \to Z$ and every distribution $\mu$ on $X$,
		\[
		\Qcc(f) \cdot \Qoip_{\mu, 1/3}(f) \cdot (\log \Qoip_{\mu, 1/3}(f) + \log |Z|) \geq \Omega(H_\text{\normalfont max}^{1/2}(\mu)).
		\]
	\end{theorem}
	Toward proving \cref{thm:distributional-tradeoff}, we begin with a Holevo-style theorem by Nayak and Salzman \cite{ns02}. Nayak and Salzman only formally stated this theorem for the case that $\mu$ is uniform, where their result gives a better dependence on the error $\epsilon$ than \cite[Theorem 1]{cdnt13}. Nayak and Salzman alluded to the non-uniform case, which follows easily from their analysis.
	\begin{theorem}[{\cite{ns02}}] \label{thm:ns02}
		Let $\mu$ be a distribution on a finite set $X$. Suppose there is a protocol that allows Alice to send Bob an element $x \sim \mu$ using an entanglement-assisted two-way qubit channel with failure probability $\epsilon$. (The failure probability is over the randomness of $x$ and the randomness of the protocol.) Then the number of qubits sent from Alice to Bob in this protocol is at least $\frac{1}{2} H_\text{\normalfont max}^{\epsilon}(\mu)$.
	\end{theorem}
	
	\begin{proof}[Proof sketch]
		Let $\gamma_x$ be the probability that the protocol succeeds in the case that Alice is sending $x$, so that $\sum_{x \in X} \mu(x) \gamma_x = 1 - \epsilon$. Say Alice sends $m$ qubits to Bob in the protocol. If $2^{2m} > |X|$, we're done, because $H_\text{max}^{\epsilon}(\mu) \leq \log |X|$. Otherwise, let $S$ be the set of the $2^{2m}$ most likely values of $x$. It is shown in the proof of \cite[Theorem 1.1]{ns02} that $\sum_{x \in X} \gamma_x \leq 2^{2m}$, so $\sum_{x \in X} \mu(x) \gamma_x \leq \sum_{x \in S} \mu(x)$. Therefore, $\sum_{x \in S} \mu(x) \geq 1 - \epsilon$, and hence $H_\text{max}^{\epsilon}(\mu) \leq \log |S| = 2m$.
	\end{proof}

	\cref{thm:ns02} is optimal: using superdense coding \cite{bw92}, Alice can transmit $x \sim \mu$ to Bob with $\epsilon$ probability of failure by sending $\lceil \frac{1}{2} H_\text{max}^{\epsilon}(\mu) \rceil$ qubits.
	
	Next, we generalize \cref{lem:clean} to show that when $\epsilon > 0$, every $f$ has a communication protocol that is \emph{approximately} clean in which Alice sends Bob at most $\Qcc_{\epsilon}(f)$ qubits. Again, this mimics the analysis in \cite{cdnt13}.
	\begin{lem} \label{lem:approximately-clean}
		For any function $f: X \times Y \to Z$ and any $\epsilon \geq 0$, there is a communication protocol $U$ in which Alice sends at most $\Qcc_{\epsilon}(f)$ qubits to Bob with the following correctness guarantee. Let $\ket{\phi}$ be the shared state in $U$, and define $\ket{\normalfont \text{error}_{x, y, a}}$ to be the (non-normalized) vector such that
		\[
			U\ket{x} \ket{y} \ket{\phi} \ket{a} = \ket{x} \ket{y} \ket{\phi} \ket{a + f(x, y)} + \ket{\normalfont \text{error}_{x, y, a}}.
		\]
		Then $\|\ket{\normalfont \text{error}_{x, y, a}}\|_2 \leq 2 |Z| \sqrt{\epsilon}$, and for each fixed $x \in X, a \in Z$, as $y$ varies, the states $\ket{\normalfont \text{error}_{x, y, a}}$ are orthogonal.
	\end{lem}

	\begin{proof}
		The \emph{construction} of $U$ is the same as that used to prove \cref{lem:clean}, with one modification: Alice and Bob begin by copying their inputs into ancilla registers, and then they never touch their input registers. Now, we analyze $U$. After executing the protocol $U_0$ for $f$, the state is
		\[
			\ket{x} \ket{y} \left(\sum_{z \in Z} \alpha_{x, y, z} \ket{A_{x, y, z}} \ket{z}\right) \ket{a},
		\]
		where $\sum_z |\alpha_{x, y, z}|^2 = 1$ and $|\alpha_{x, y, f(x, y)}|^2 \geq 1 - \epsilon$. After applying the unitary $V$ defined in the proof of \cref{lem:clean}, the state is
		\[
			\ket{x} \ket{y} \left(\sum_{z \in Z} \alpha_{x, y, z} \ket{A_{x, y, z}} \ket{z} \ket{a + z}\right),
		\]
		which we can rewrite as
		\[
			\ket{x} \ket{y} \left(\sum_{z \in Z} \alpha_{x, y, z} \ket{A_{x, y, z}} \ket{z}\right) \ket{a + f(x, y)} + \ket{x} \ket{y} \left(\sum_{\substack{z \in Z \\ z \neq f(x, y)}} \alpha_{x, y, z} \ket{A_{x, y, z}} \ket{z} (\ket{a + z} - \ket{a + f(x, y)})\right).
		\]
		It follows that
		\[
			\ket{\text{error}_{x, y, a}} = (U_0^{-1} \otimes I) \ket{x} \ket{y}\left(\sum_{\substack{z \in Z \\ z \neq f(x, y)}} \alpha_{x, y, z} \ket{A_{x, y, z}} (\ket{a + z} - \ket{a + f(x, y)})\right).
		\]
		These vectors are othogonal as $y$ varies, because $\ket{y}$ is a factor of $(U_0 \otimes I) \ket{\text{error}_{x, y, a}}$. We now bound the $\ell_2$ norm:
		\begin{align*}
			\|\ket{\text{error}_{x, y, a}}\|_2 &= \left\|\ket{x} \ket{y}\left(\sum_{\substack{z \in Z \\ z \neq f(x, y)}} \alpha_z \ket{A_{x, y, z}} (\ket{a + z} - \ket{a + f(x, y)})\right)\right\|_2 \quad \text{by unitarity} \\
			&\leq 2\sum_{\substack{z \in Z \\ z \neq f(x, y)}} |\alpha_z| \quad \text{by the triangle inequality} \\
			&\leq 2|Z| \sqrt{\epsilon}. \qedhere
		\end{align*}
	\end{proof}

	We rely on a standard lemma relating $\ell_2$ distance and total variation distance:
	\begin{lem}[{\cite{bv97}}] \label{lem:variation-distance}
		Suppose $\ket{\phi}, \ket{\psi}$ are pure quantum states of the same dimension satisfying $\|\ket{\phi} - \ket{\psi}\|_2 \leq \epsilon$. Then the total variation distance between the two probability distributions resulting from making the same measurement on $\ket{\phi}$ vs. $\ket{\psi}$ is at most $4\epsilon$.
	\end{lem}

	Now, like in the proof of \cref{thm:zero-error}, we use the approximately clean protocol to transmit data from Alice to Bob:
	\begin{lem} \label{lem:transmit-data}
		Fix any $f: X \times Y \to Z$, any distribution $\mu$ over $X$, and any $\epsilon > 0$. There is a protocol for sending $x \sim \mu$ from Alice to Bob with failure probability at most
		\[
			\frac{1}{3} + O(\sqrt{\epsilon} |Z|^2 \Qoip_{\mu, 1/3}(f))
		\]
		using an entanglement-assisted qubit channel in which Alice sends at most $\Qcc_{\epsilon}(f) \cdot \Qoip_{\mu, 1/3}(f)$ qubits to Bob. The failure probability is over the randomness of $x$ and the internal randomness of the protocol.
	\end{lem}

	\begin{proof}
		As in the proof of \cref{thm:zero-error}, Bob executes the $\Qoip_{\mu, 1/3}(f)$ protocol. He simulates each query using the communication protocol of \cref{lem:approximately-clean}. The total number of qubits sent from Alice to Bob in this protocol is $\Qcc(f) \cdot \Qoip_{\mu, 1/3}(f)$.
		
		Now we analyze the failure probability. In this protocol, the state is initially
		\[
			\ket{x} \ket{\phi} \ket{0} \ket{0} \ket{0} \ket{0}.
		\]
		Here, we have ordered the registers as follows: Alice's input register; the shared register for the communication protocol; a register which is simultaneously Bob's input register for the communication protocol and the query register for the OIP algorithm; a register which is simultaneously the output register of the communication protocol and of the OIP oracle; the ancilla state of the query algorithm; the output register of the query algorithm. We show by induction on $i$ that after $i$ queries, the state is $(2|Z|^2 \sqrt{\epsilon} i)$-close (in $\ell_2$ distance) to the state
		\[
			\ket{x} \ket{\phi} \ket{A_{x, i}},
		\]
		where $\ket{A_{x, i}}$ is the state that the $\Qoip_{\mu, 1/3}(f)$ algorithm is in after making $i$ queries. For the base case, when $i = 0$, this is trivially true. For the inductive step, assume it is true for some value of $i$. Then Bob will apply a unitary operation $V_i$ for processing before query $i + 1$. After this, by unitarity, the state is $(2|Z|^2 \sqrt{\epsilon} i)$-close to
		\[
			\ket{x} \ket{\phi} \left(\sum_{y \in Y, a \in Z} \alpha_{y, z, i + 1} \ket{y} \ket{a} \ket{B_{x, y, a, i + 1}}\right),
		\]
		where $\sum_{y \in Y, a \in Z} \alpha_{y, z, i + 1} \ket{y} \ket{a} \ket{B_{x, y, a, i + 1}}$ is the state that the $\Qoip_{\mu, 1/3}(f)$ algorithm is in just before making query $i + 1$. Then, Bob runs the communication protocol, leading to a state that is $(2|Z|^2 \sqrt{\epsilon} i)$-close to
		\[
			\ket{x} \ket{\phi} \left(\sum_{y \in Y, a \in Z} \alpha_{y, z, i + 1}\ket{y} \ket{a + f(x, y)} \ket{B_{x, y, a, i + 1}}\right) + \sum_{y \in Y, a \in Z} \alpha_{y, z, i + 1} \ket{\text{error}_{x, y, a}} \ket{B_{x, y, a, i + 1}}.
		\]
		Now, of course, $\sum_{y \in Y, a \in Z} \alpha_{y, z, i + 1}\ket{y} \ket{a + f(x, y)} \ket{B_{x, y, a, i + 1}}$ is another way of writing $A_{x, i + 1}$. Furthermore,
		\begin{align*}
			\left\|\sum_{y \in Y, a \in Z} \alpha_{y, z, i + 1} \ket{\text{error}_{x, y, a}} \ket{B_{x, y, a, i + 1}}\right\|_2 &\leq \sum_{a \in Z} \left\| \sum_{y \in Y} \alpha_{y, z, i + 1} \ket{\text{error}_{x, y, a}} \ket{B_{x, y, a, i + 1}}\right\|_2
		\end{align*}
		by the triangle inequality. Since the states $\ket{\text{error}_{x, y, a}}$ are orthogonal as $y$ varies and each $\|\ket{\text{error}_{x, y, a}}\|_2$ is bounded by $2|Z| \sqrt{\epsilon}$, we have
		\begin{align*}
			\left\|\sum_{y \in Y, a \in Z} \alpha_{y, z, i + 1} \ket{\text{error}_{x, y, a}} \ket{B_{x, y, a, i + 1}}\right\|_2 &\leq \sum_{a \in Z} 2|Z| \sqrt{\epsilon} \cdot  \sqrt{\sum_{y \in Y} |\alpha_{y, z, i + 1}|^2} \\
			&\leq \sum_{a \in Z} 2 |Z| \sqrt{\epsilon} \\
			&\leq 2|Z|^2 \sqrt{\epsilon}.
		\end{align*}
		Another application of the triangle inequality completes the induction. It follows that just before measurement, the state is $(2|Z|^2 \sqrt{\epsilon} \Qoip_{\mu, 1/3}(f))$-close to what it would be in the $\Qoip_{\mu, 1/3}(f)$ protocol.
		
		Let $\gamma_x$ be the probability that the $\Qoip_{\mu, 1/3}(f)$ protocol fails in the case that the correct answer is $x \in X$, so that $\E_{x \sim \mu}[\gamma_x]$ is the overall failure probability of the $\Qoip_{\mu, 1/3}(f)$ algorithm and hence is bounded by $1/3$. By \cref{lem:variation-distance}, the probability that our data transmission protocol fails in the case that Alice holds $x$ is at most $\gamma_x + 8|Z|^2 \sqrt{\epsilon} \Qoip_{\mu, 1/3}(f)$. Therefore, the overall failure probability of our data transmission protocol is bounded by
		\[
			\E_{x \sim \mu}[\gamma_x + 8|Z|^2 \sqrt{\epsilon} \Qoip_{\mu, 1/3}(f)] \leq 1/3 + O(\sqrt{\epsilon} |Z|^2 \Qoip_{\mu, 1/3}(f)). \qedhere
		\]
	\end{proof}
	
	\begin{lem} \label{lem:distributional-tradeoff-before-amplification}
		For every $f: X \times Y \to Z$ and every distribution $\mu$ on $X$, there is a value $\epsilon \geq \Omega(|Z|^{-4} \Qoip_{\mu, 1/3}(f)^{-2})$ such that
		\[
			\Qcc_{\epsilon}(f) \cdot \Qoip_{\mu, 1/3}(f) \geq \frac{1}{2} H_\text{\normalfont max}^{1/2}(\mu).
		\]
	\end{lem}

	\begin{proof}
		Choose $\epsilon$ small enough that the failure probability in \cref{lem:transmit-data} is only $1/2$. Then apply \cref{thm:ns02}.
	\end{proof}

	\begin{proof}[Proof of \cref{thm:distributional-tradeoff}]
		Let $\epsilon$ be the value in \cref{lem:distributional-tradeoff-before-amplification}. By straightforward amplification,
		\[
			\Qcc_{\epsilon}(f) \leq O(\Qcc(f) \cdot (\log \Qoip(f) + \log |Z|).
		\]
		Therefore,
		\[
			\Qcc(f) \cdot \Qoip_{\mu, 1/3}(f) \cdot (\log \Qoip_{\mu, 1/3}(f) + \log |Z|) \geq \Omega(\Qcc_{\epsilon}(f)) \cdot \Qoip_{\mu, 1/3}(f) \geq \Omega(H_\text{max}^{1/2}(\mu)). \qedhere
		\]
	\end{proof}

	\begin{proof}[Proof of \cref{thm:main}]
		Let $\mu$ be the uniform distribution on $X$. Then $H_\text{max}^{1/2}(\mu) = \log |X| - 1$ and $\Qoip(f) \geq \Qoip_{\mu, 1/3}(f)$, so applying \cref{thm:distributional-tradeoff} completes the proof.
	\end{proof}

	\section{Applications} \label{sec:applications}

	\subsection{Determining whether $x + y$ is a perfect square}
	Fix a finite field $\F_q$. Define $\mathsf{PS}: \F_q \times \F_q \to \{0, 1\}$ by
	\[
	\mathsf{PS}(x, y) = \begin{cases}
	1 & \text{if there is some $z \in \F_q$ such that $x + y = z^2$} \\
	0 & \text{otherwise.}
	\end{cases}
	\]
	
	\begin{theorem} \label{thm:ps}
		If $q$ is a power of an odd prime, then $\Qcc(\mathsf{PS}) \geq \Omega(\log q)$.
	\end{theorem}
	
	\begin{proof}
		Define $\mathsf{\chi}: \F_q \to \{-1, 0, 1\}$ by
		\begin{equation} \label{eqn:chi}
			\chi(x) = \begin{cases}
				1 & \text{if there is some nonzero } z \in \F_q \text{ such that } x = z^2 \\
				0 & \text{if } x = 0 \\
				-1 & \text{otherwise.}
			\end{cases}
		\end{equation}
		Define $\mathsf{PS}': \F_q \times \F_q \to \{-1, 0, 1\}$ by $\mathsf{PS}'(x, y) = \chi(x + y)$. Van Dam showed \cite{dam02} that $\Qoip_0(\mathsf{PS}') \leq 2$. By \cref{thm:main}, it follows that $\Qcc(\mathsf{PS}') \geq \Omega(\log q)$. Finally, $\Qcc(\mathsf{PS}') \leq O(\Qcc(\mathsf{PS}))$, because to compute $\mathsf{PS}'(x, y)$, Alice and Bob can compute $\mathsf{PS}(x, y)$ and also use an equality protocol to check if $x = -y$.
	\end{proof}
	
	Note that the hypothesis of \cref{thm:ps} is necessary, because if $q$ is a power of two, then \emph{every} element of $\F_q$ is a perfect square. The OIP algorithm by van Dam \cite{dam02} that is used in the proof of \cref{thm:ps} is related to Paley's construction of Hadamard\footnote{Here, a \emph{Hadamard matrix} is any matrix with $\pm 1$ entries whose rows are orthogonal.} matrices \cite{pal33}, and a closely related algorithm by van Dam \cite{dam02} can be combined with \cref{thm:main} to prove that if the communication matrix associated with $f: X \times X \to \{-1, 1\}$ is a Hadamard matrix, then $\Qcc(f) \geq \Omega(\log |X|)$. (An alternate proof of this last fact: if the communication matrix associated with $f$ is a Hadamard matrix, then the discrepancy of $f$ with respect to the uniform distribution is at most $|X|^{-1/2}$ \cite{kn97}. The $\Qcc$ lower bound follows by a theorem by Linial and Shraibman \cite{ls09b}.) We now give $\Qcc$ lower bounds for two other algebraic problems.
	
	\paragraph{General multiplicative characters.}
	The function $\chi$ defined by \cref{eqn:chi} is an example of a \emph{multiplicative character}. For our purposes, a multiplicative character of $\F_q$ is a function $\chi: \F_q \to \C$ such that $\chi(xy) = \chi(x) \chi(y)$ for all $x, y \in \F_q$. We say that $\chi$ is \emph{trivial} if $\chi(x) = 1$ for all nonzero $x$. \cref{thm:ps} extends to arbitrary nontrivial multiplicative characters:
	\begin{theorem}
		Suppose $\chi: \F_q \to \C$ is a nontrivial multiplicative character, $g$ is a generator of $\F_q^{\times}$, and $d$ is a positive integer such that $\chi(g^d) = 1$. Define $f: \F_q \times \F_q \to \C$ by $f(x, y) = \chi(x + y)$. Then $\Qcc(f) \geq \Omega(\log q /\log(d + 1))$.
	\end{theorem}

	\begin{proof}
		The image $\chi(\F_q)$ has at most $d + 1$ elements, namely $\chi(g), \chi(g^2), \dots, \chi(g^d)$, and $\chi(0)$. Therefore, we can think of $f$ as a function $\F_q \times \F_q \to Z$ where $|Z| \leq d + 1$. Finally, van Dam, Hallgren, and Ip \cite{dhi06} showed that $\Qoip(f) = 1$. Applying \cref{thm:main} completes the proof.
	\end{proof}
	
	\paragraph{Polynomials and perfect squares.} The $\mathsf{PS}'$ function that appears in the proof of \cref{thm:ps} is a special case of the following construction studied by Russell and Shparlinski \cite{rs04}. For $\mathcal{C} \subseteq \F_q[y]$, define $f_{\mathcal{C}}: \mathcal{C} \times \F_q \to \{-1, 0, 1\}$ by $f_{\mathcal{C}}(p, y) = \chi(p(y))$, where $\chi$ is defined by \cref{eqn:chi}. (So $\mathsf{PS}'$ is the case $\mathcal{C} = \{x + y : x \in \F_q\}$.) The following theorem is similar to but incomparable with \cref{thm:ps}, since it assumes that $q$ is prime:
	\begin{theorem}
		Assume $q$ is an odd prime. Suppose $\mathcal{C}$ is a set of monic, square-free polynomials of degree at most $d$. Then
		\[
		\Qcc(f_{\mathcal{C}}) \geq \Omega\left(\frac{\log |\mathcal{C}|}{d \log (d + 1)}\right).
		\]
	\end{theorem}
	
	\begin{proof}
		Russell and Shparlinski showed under the specified assumptions \cite{rs04} that $\Qoip(f_{\mathcal{C}}) \leq O(d)$. The theorem follows immediately by applying \cref{thm:main}.
	\end{proof}
	
	\subsection{Ordered search} \label{sec:gt}
	
	\subsubsection{Background}
	Fix $N \in \N$. In the ordered search problem, there is an unknown number $x \in [N]$. An oracle answers queries of the form ``Is $y < x$?'' for a specified $y \in [N]$. The goal is to find $x$. In other words, ordered search is the oracle identification problem associated with the function $\mathsf{GT}: [N] \times [N] \to \{0, 1\}$ defined by
	\begin{equation} \label{eqn:gt-def}
	\mathsf{GT}(x, y) = \begin{cases}
	1 & \text{if } x > y \\
	0 & \text{if } x \leq y.
	\end{cases}
	\end{equation}
	
	The deterministic complexity of this problem is $\lceil \log N \rceil$ by binary search. The first lower bound on the \emph{quantum} complexity of ordered search was by Buhrman and de Wolf \cite{bw99}, who showed that $\Qoip(\mathsf{GT}) \geq \Omega(\sqrt{\log N} / \log \log N)$. The bound was subsequently improved by Farhi et al. \cite{fggs98a} to $\log N / (2 \log \log N)$, and then further improved by Ambainis \cite{amb99} to $\Omega(\log N)$. This bound is of course asymptotically optimal, but the leading constant is still interesting. H{\o}yer et al. \cite{hns02} showed that $\Qoip_0(\mathsf{GT}) \geq \frac{1}{\pi} \ln N - O(1)$, which is currently the best known lower bound for \emph{exact} algorithms. In the case $\epsilon > 0$, Ben-Or and Hassidim \cite{boh07} showed that $\Qoip_{\epsilon}(\mathsf{GT}) \geq \frac{1 - \epsilon}{\pi} \ln N - O(\epsilon)$, which has better $\epsilon$ dependence than the bound by H{\o}yer et al. A different line of work \cite{fggs99, hns02, jlb05, clp07, boh07} has investigated \emph{upper} bounds on $\Qoip(\mathsf{GT})$.
	
	\subsubsection{The new lower bound argument}
	In this section, we give a new proof that $\Qoip(\mathsf{GT}) \geq \Omega(\log N / \log \log N)$. Our proof is much simpler than previous arguments, and it gives a new perspective on the complexity of quantum ordered search.
	\begin{theorem} \label{thm:gt}
		$\Qoip(\mathsf{GT}) \geq \Omega(\log N / \log \log N)$.
	\end{theorem}
	
	\begin{proof}
		Nisan \cite{nis93} gave a classical $\epsilon$-error public-coin protocol for $\mathsf{GT}$ in which Alice and Bob exchange $O(\log \log N + \log(1/\epsilon))$ bits. In particular, this implies that $\Qcc_{\epsilon}(\mathsf{GT}) \leq O(\log \log N + \log(1/\epsilon))$. Let $\epsilon$ be the value in \cref{lem:distributional-tradeoff-before-amplification}. Then \cref{lem:distributional-tradeoff-before-amplification} implies that there is some constant $c \in (0, 2)$ (independent of everything) such that
		\begin{equation} \label{eqn:gt}
		(\log \log N + \log \Qoip(\mathsf{GT})) \cdot \Qoip(\mathsf{GT}) \geq c \cdot \log N.
		\end{equation}
		Assume for a contradiction that
		\begin{equation} \label{eqn:gt-cont}
		\Qoip(\mathsf{GT}) < \frac{(c/2) \cdot \log N}{\log \log N}.
		\end{equation}
		The denominator of the right-hand side of \cref{eqn:gt-cont} is at least $1$ (for sufficiently large $N$) and $c/2 < 1$, so $\Qoip(\mathsf{GT}) < \log N$. Therefore, $\log \Qoip(\mathsf{GT}) < \log \log N$. By \cref{eqn:gt}, this contradicts \cref{eqn:gt-cont}.
	\end{proof}
	
	One might hope to prove a better lower bound on $\Qoip(\mathsf{GT})$ by using a better quantum communication protocol for $\mathsf{GT}$. Alas, this is not possible: Braverman and Weinstein \cite{bw16} showed that there is a distribution with respect to which the discrepancy of $\mathsf{GT}$ is $O(1/\sqrt{\log N})$, and Linial and Shraibman \cite{ls09b} showed that the discrepancy method applies to $\Qcc$, so $\Qcc(\mathsf{GT}) \geq \Omega(\log \log N)$.
	
	\subsubsection{Distributional complexity}
	Since our tradeoff theorem extends to \emph{distributional} query complexity, the proof of \cref{thm:gt} can be extended to show that $\Qoip_{\mu, 1/3}(\mathsf{GT}) \geq \Omega(H_\text{max}^{1/2}(\mu) / \log \log N)$ for any distribution $\mu$ on $[N]$. At first glance, this seems like a new, interesting lower bound. But in fact, a superior bound of $\Omega(H_\text{max}^{1/2}(\mu))$ follows easily from Ambainis's result that $\Qoip(\mathsf{GT}) \geq \Omega(\log N)$ because of the special structure of the $\mathsf{GT}$ function. To show this, we need the following elementary fact: if $f: X \to \{0, 1\}$ is a function and $U_f$ is the oracle $U_f \ket{x} \ket{a} = \ket{x} \ket{a \oplus f(x)}$, then a \emph{controlled} $U_f$ operator can be implemented using two queries to $U_f$, a Toffoli gate, and one ancilla qubit (initially $\ket{0}$ and restored to $\ket{0}$). For completeness, we give the circuit in \cref{fig:control}.
	
	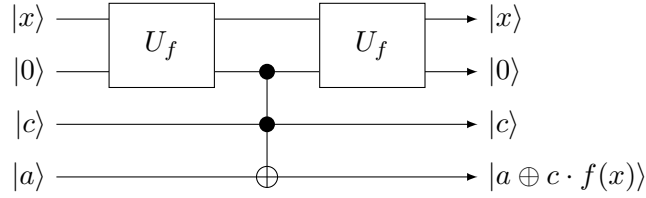
\begin{figure}
		\centering
		\begin{tikzpicture}[scale=0.7]
		\node[anchor=east] (a) at (0, 0) {$\ket{x}$};
		\node[anchor=east] (b) at (0, -1) {$\ket{0}$};
		\node[anchor=east] (c) at (0, -2) {$\ket{c}$};
		\node[anchor=east] (d) at (0, -3) {$\ket{a}$};
		
		\node[anchor=west] (a') at (8, 0) {$\ket{x}$};
		\node[anchor=west] (b') at (8, -1) {$\ket{0}$};
		\node[anchor=west] (c') at (8, -2) {$\ket{c}$};
		\node[anchor=west] (d') at (8, -3) {$\ket{a \oplus c \cdot f(x)}$};
		
		\draw[->, >=latex] (a) -- (a');
		\draw[->, >=latex] (b) -- (b');
		\draw[->, >=latex] (c) -- (c');
		\draw[->, >=latex] (d) -- (d');
		
		\filldraw[draw=black, fill=white] (1, 0.3) rectangle node {$U_f$} (3, -1.3);
		
		\fill (4, -1) circle (0.15);
		\fill (4, -2) circle (0.15);
		\draw (4, -3) circle (0.2);
		\draw (4, -1) -- (4, -3.2);
		
		\filldraw[draw=black, fill=white] (5, 0.3) rectangle node {$U_f$} (7, -1.3);
		\end{tikzpicture}
		\caption{An implementation of controlled-$U_f$. The third register, initialized $\ket{c}$, is the control.} \label{fig:control}
	\end{figure}
	
	\begin{prop} \label{prop:gt}
		For any distribution $\mu$ on $[N]$, $\Qoip_{\mu, 1/3}(\mathsf{GT}) \geq \Omega(H_\text{max}^{1/2}(\mu))$.
	\end{prop}

	\begin{proof}
		Fix some $\Qoip_{\mu, 1/3}(\mathsf{GT})$ algorithm. Let $\epsilon_x$ be the failure probability of the algorithm in the case that the correct output is $x \in [N]$. Then $\E_{x \sim \mu}[\epsilon_x] \leq 1/3$. By Markov's inequality, $\Pr_{x \sim \mu}[\epsilon_x > 2/3] \leq 1/2$. Let $S = \{x : \epsilon_x \leq 2/3\}$, so that $|S| \geq 2^{H_\text{max}^{1/2}(\mu)}$. Write $S = \{s_1, s_2, \dots, s_{N'}\}$, where $s_i < s_{i + 1}$.
		
		Let $\mathsf{GT}'$ be the restriction of $\mathsf{GT}$ to $[N'] \times [N']$. We will show that $\Qoip_{2/3}(\mathsf{GT}') \leq 2 \Qoip_{\mu, 1/3}(\mathsf{GT})$. The idea is simple: ordered search of $[N']$ is equivalent (by relabeling) to ordered search of $S$, which is achieved by the $\Qoip_{\mu, 1/3}(\mathsf{GT})$ algorithm. We give details for completeness.
		
		Define $A: [N] \to \{0, 1\}$ by
		\[
		A(y) = \begin{cases}
		1 & \text{if } s_1 \leq y \\
		0 & \text{otherwise.}
		\end{cases}
		\]
		Define $B: [N] \to [N']$ by
		\[
			B(y) = \begin{cases}
				\max\{i: s_i \leq y\} & \text{if } s_1 \leq y \\
				1 & \text{otherwise.}
			\end{cases}
		\]
		Let $V$ be the unitary operation on $\C^N \otimes \C^2 \otimes \C^{N'}$ defined by
		\begin{align*}
			V\ket{y} \ket{a} \ket{b} &= \ket{y} \ket{a \oplus A(y)} \ket{b + B(y) \text{ mod } N'}.
		\end{align*}
		The algorithm for ordered search of $[N']$: Run the $\Qoip_{\mu, 1/3}(\mathsf{GT})$ algorithm, but when it tries to make a query, instead apply the circuit in \cref{fig:gt}. When it outputs a value $\hat{x}$, if $\hat{x} = s_j$ for some $j$, output $j$.
		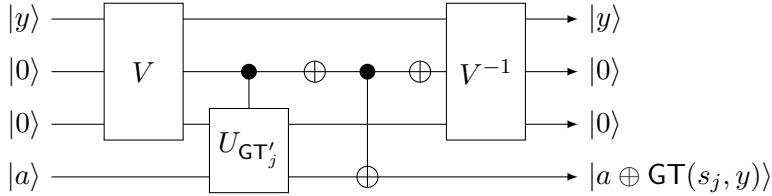
\begin{figure}
			\centering
			\begin{tikzpicture}[scale=0.7]
			\node[anchor=east] (a) at (0, 0) {$\ket{y}$};
			\node[anchor=east] (b) at (0, -1) {$\ket{0}$};
			\node[anchor=east] (c) at (0, -2) {$\ket{0}$};
			\node[anchor=east] (d) at (0, -3) {$\ket{a}$};
			
			\node[anchor=west] (a') at (10, 0) {$\ket{y}$};
			\node[anchor=west] (b') at (10, -1) {$\ket{0}$};
			\node[anchor=west] (c') at (10, -2) {$\ket{0}$};
			\node[anchor=west] (d') at (10, -3) {$\ket{a \oplus \mathsf{GT}(s_j, y)}$};
			
			\draw[->, >=latex] (a) -- (a');
			\draw[->, >=latex] (b) -- (b');
			\draw[->, >=latex] (c) -- (c');
			\draw[->, >=latex] (d) -- (d');
			
			\filldraw[draw=black, fill=white] (1, 0.3) rectangle node {$V$} (2.5, -2.3);
			
			\filldraw[draw=black, fill=white] (3, -1.7) rectangle node {$U_{\mathsf{GT}'_j}$} (4.5, -3.3);
			
			\filldraw[draw=black, fill=white] (7.5, 0.3) rectangle node {$V^{-1}$} (9, -2.3);
			
			\fill (3.75, -1) circle (0.15);
			\fill (6, -1) circle (0.15);
			\draw (5, -1) circle (0.2);
			\draw (6, -3) circle (0.2);
			\draw (7, -1) circle (0.2);
			\draw (3.75, -1) -- (3.75, -1.7);
			\draw (5, -0.8) -- (5, -1.2);
			\draw (6, -1) -- (6, -3.2);
			\draw (7, -0.8) -- (7, -1.2);
			
			\end{tikzpicture}
			\caption{Simulating a query for the ordered search of $[N]$ using an oracle for the ordered search of $[N']$. Here, $U_{\mathsf{GT}'_j}$ is the oracle for ordered search of the unknown value $j \in [N']$, i.e. $U_{\mathsf{GT}'_j} \ket{i} \ket{a} = \ket{i} \ket{a \oplus \mathsf{GT}'(j, i)}$.} \label{fig:gt}
		\end{figure}
		
		By the construction of $V$, the circuit in \cref{fig:gt} simulates an oracle for searching for $s_j \in [N]$ when the true oracle is for searching for $j \in [N']$. Therefore, with probability at least $1 - \epsilon_{s_j}$, i.e. with probability at least $1/3$, the simulation of the $\Qoip_{\mu, 1/3}(\mathsf{GT})$ algorithm will output $s_j$.
		
		Two queries suffice to implement the circuit of \cref{fig:gt}, so $\Qoip_{2/3}(\mathsf{GT}') \leq 2\Qoip_{\mu, 1/3}(\mathsf{GT})$. Ambainis showed \cite{amb99} that $\Qoip(\mathsf{GT}') \geq \Omega(\log N')$, and hence (by amplification) $\Qoip_{2/3}(\mathsf{GT}') \geq \Omega(\log N') \geq \Omega(H_\text{max}^{1/2}(\mu))$. Therefore, $\Qoip_{\mu, 1/3}(\mathsf{GT}_N) \geq \Omega(H_\text{max}^{1/2}(\mu))$.
	\end{proof}

	Note that the $\epsilon$-error \emph{classical} distributional complexity of ordered search is exactly $\lceil H_\text{max}^{\epsilon}(\mu) \rceil$.
	
	\section{Near-optimality of our tradeoff theorems} \label{sec:optimality}
	\subsection{The exact case}
	In this section, we show that \cref{thm:zero-error} is close to the best possible tradeoff in terms of only $\Qcc_0(f)$, $\Qoip_0(f)$, and $|X|$. There are examples, such as $\mathsf{IP}$, where \cref{thm:zero-error} is exactly tight \cite{cdnt13}. But the point is that we are giving a \emph{family} of examples where \cref{thm:zero-error} is close to tight and $\Qoip_0(f)$ takes on \emph{all possible values} up to $\log |X|$:
	
	\begin{theorem} \label{thm:zero-error-optimality}
		For any finite set $X$ and any positive integer $q \leq \log |X|$, there exists a finite set $Y$ and a function $f: X \times Y \to \{0, 1\}$ such that $\Qoip_0(f) = q$ and
		\[
			\Qoip_0(f) \cdot \Qcc_0(f) \leq \frac{1}{2} \log |X| +  O(\Qoip_0(f) \cdot \log \Qoip_0(f)).
		\]
	\end{theorem}
	
	\begin{figure}
		\centering
		\begin{tikzpicture}[scale=12]
			\xdef\hs{50}
			\pgfmathparse{2*\hs}
			\xdef\s{\pgfmathresult}
			\pgfmathparse{2.0/\s}
			\xdef\b{\pgfmathresult}
			
			\draw[smooth, samples=200, domain=\b:1, thick] plot(\x, {\b / \x});
			\fill[smooth, samples=200, domain=\b:1, pattern=north west lines] plot(\x, {\b / \x}) -- (1, 0) -- (0, 0) -- (0, 1);
			
			\xdef\c{}
			
			\foreach \q in {1, ..., \hs}
			{
				\pgfmathparse{\b * \q}
				\xdef\x{\pgfmathresult}
				\pgfmathparse{ceil(\s / \q)}
				\xdef\n{\pgfmathresult}
				\pgfmathparse{ceil(\n/2) + ceil(log2(\q)/2)}
				\xdef\qcc{\pgfmathresult}
				\pgfmathparse{\qcc*\b}
				\xdef\y{\pgfmathresult}
				
				\node at (\x, \y) {$\bullet$};
				
				\pgfmathparse{ceil(\hs / \q)}
				\xdef\bound{\pgfmathresult}
				\pgfmathparse{\bound * \b}
				\xdef\boundy{\pgfmathresult}
				
				\node at (\x, \boundy) {$\times$};
				
			}
		
		
			\draw[thick, <->, >=latex] (0, 1.05) -- node[left] {$\Qcc_0$} (0, 0) -- node[below] {$\Qoip_0$} (1.05, 0);
			\draw (-0.02, 1) node[left] {$50$} -- (0, 1);
			\draw (-0.02, \b) node[left] {$1$} -- (0, \b);
			\draw (1, -0.02) node[below] {$50$} -- (1, 0);
			\draw (\b, -0.02) node[below] {$1$} -- (\b, 0);
		\end{tikzpicture}
		\caption{The tradeoff for exact protocols/algorithms when $X = \{0, 1\}^{100}$. \cref{thm:zero-error} implies that every function $f: X \times Y \to Z$ lies on or above the curve, i.e. outside the shaded region. For each value of $\Qoip_0(f)$, the $\times$ symbol marks the smallest integer value of $\Qcc_0(f)$ consistent with \cref{thm:zero-error}, and the $\bullet$ symbol marks the function constructed in the proof of \cref{thm:zero-error-optimality}.} \label{fig:tradeoff}
	\end{figure}
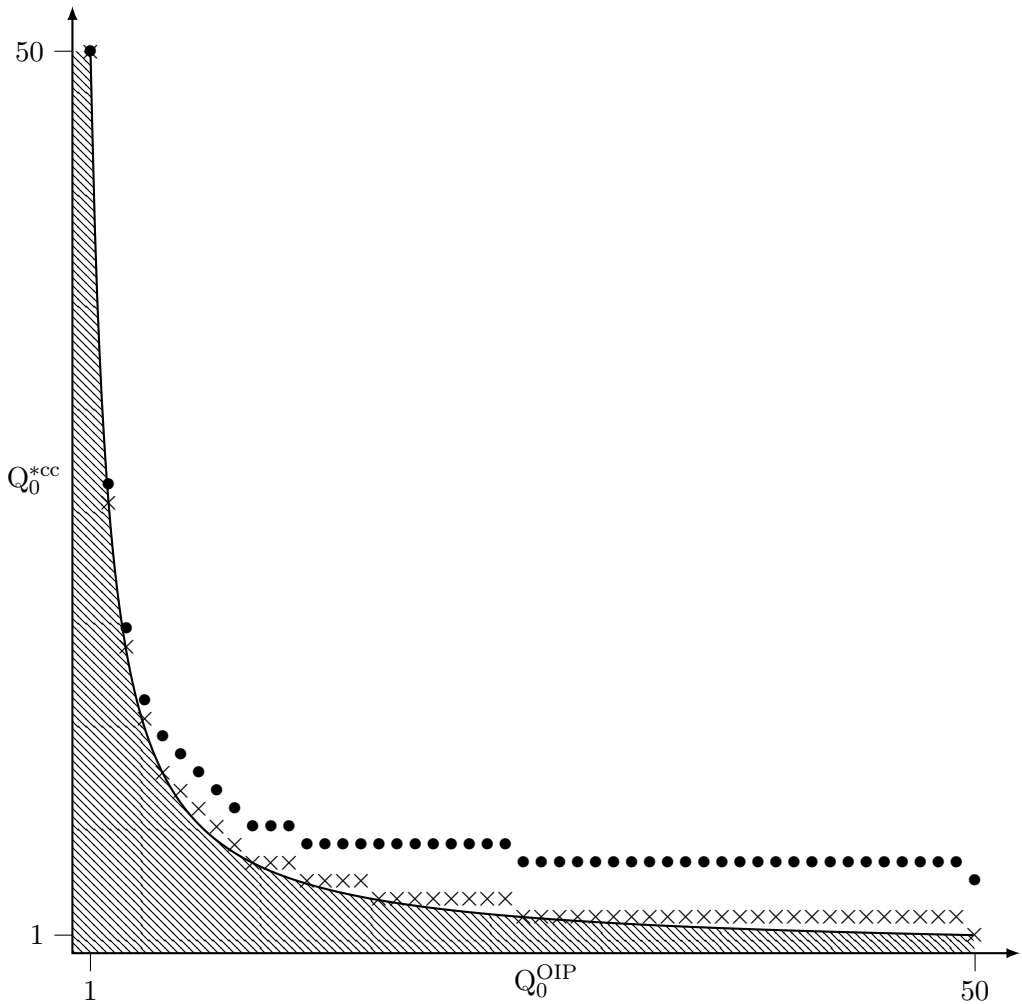
	
	The proof of \cref{thm:zero-error-optimality} is based on the fact that $\mathsf{IP}$ and $\mathsf{INDEX}$ approximately achieve the two extremes of \cref{thm:zero-error}; we interpolate by appropriately \emph{composing} $\mathsf{IP}$ with $\mathsf{INDEX}$. To lower bound $\Qoip_0(f)$, we use the following lemma by Beals et al. on the quantum query complexity of the $\mathsf{OR}$ function with a controlled\footnote{Formally, Beals et al. only proved the simpler version of \cref{lem:or} with a non-controlled oracle. But as they observed in a footnote, the proof generalizes to the controlled case.} oracle:
	\begin{lem}[{\cite{bbc+01}}] \label{lem:or}
		Suppose $h: [q] \to \{0, 1\}$ is an unknown function. Let $U$ be the unitary operator on $\C^q \otimes \C^2 \otimes \C^2$ defined by
		\[
			U \ket{j} \ket{c} \ket{a} = \ket{j} \ket{c} \ket{a \oplus c \cdot h(j)}.
		\]
		Then $q$ queries to $U$ are required to exactly compute $\mathsf{OR}(h(1), \dots, h(q))$.
	\end{lem}
	
	(It will not matter for us that the $\mathsf{OR}$ function in particular appears in \cref{lem:or}. All that we will actually need is that $q$ queries are required to learn the entire function $h(\cdot)$.)

	\begin{proof}[Proof of \cref{thm:zero-error-optimality}]
		Define $n = \lceil \frac{1}{q} \log |X|\rceil$. Define $X_0 = \{0, 1\}^{nq}$ and $Y = [q] \times \{0, 1\}^n$. We first define a function $g: X_0 \times Y \to \{0, 1\}$. For any $x_1, \dots, x_n \in \{0, 1\}^q, j \in [q], y \in \{0, 1\}^n$, we define
		\[
			g((x_1, \dots, x_n), (j, y)) = \sum_{i = 1}^n y_i \cdot x_{ij} \pmod{2},
		\]
		where $x_{ij}$ denotes the $j$th bit of $x_i$. In other words, the output of $g$ is $\mathsf{IP}(w, y)$, where $w_i = \mathsf{INDEX}(x_i, j)$. Identify $X$ with some subset of $X_0$ of size $|X|$ such that $(x_1, 0^q, 0^q, \dots, 0^q) \in X$ for all $x_1 \in \{0, 1\}^q$. (This is possible because $q \leq \log |X|$.) Let $f$ be the restriction of $g$ to $X \times Y$.
		
		We first show that $\Qoip_0(f) \leq q$. In short, to compute $x$, for $j = 1$ to $q$, we learn $x_{ij}$ for every $i$ using a single query via the Bernstein-Vazirani algorithm. For completeness, we now give the algorithm in explicit detail.
		\begin{enumerate}
			\item For $j = 1$ to $q$:
			\begin{enumerate}
				\item Prepare the state
				\[
					\ket{\phi_0} = \frac{1}{\sqrt{2^{n + 1}}}\sum_{\substack{y \in \{0, 1\}^n \\ z \in \{0, 1\}}} (-1)^z \ket{j} \ket{y} \ket{z}.
				\]
				\item Query the $f_x$ oracle, giving the state
				\begin{align*}
					\ket{\phi_1} &= \frac{1}{\sqrt{2^{n + 1}}} \sum_{\substack{y \in \{0, 1\}^n \\ z \in \{0, 1\}}} (-1)^z \ket{j} \ket{y} \ket{z \oplus g((x_1, \dots, x_n), (j, y)} \\
					&= \frac{1}{\sqrt{2^{n + 1}}} \sum_{\substack{w \in \{0, 1\}^n \\ z' \in \{0, 1\}}} (-1)^{z' + g((x_1, \dots, x_n), (j, y))} \ket{j} \ket{y} \ket{z'} \\
					&= \frac{1}{\sqrt{2^{n + 1}}} \sum_{\substack{w \in \{0, 1\}^n \\ z' \in \{0, 1\}}} (-1)^{z' + \sum_{i = 1}^n y_i \cdot x_{ij}} \ket{j} \ket{y} \ket{z'}.
				\end{align*}
				\item Apply a Hadamard transformation to each of the last $n + 1$ qubits, giving the state
				\begin{align*}
					\ket{\phi_2} &= \ket{j} \ket{x_{1j}, \dots, x_{nj}} \ket{1}.
				\end{align*}
				\item Measure to learn $x_{1j}, \dots, x_{nj}$.
			\end{enumerate}
		\end{enumerate}
		Next, we show that $\Qoip_0(f) \geq q$ by a reduction from \cref{lem:or}. Fix some function $h: [q] \to \{0, 1\}$. Define
		\[
		x_{ij} = \begin{cases}
		h(j) & \text{if } i = 1 \\
		0 & \text{if } i \neq 1.
		\end{cases}
		\]
		By our definition of $X$, $x \in X$. A query to $f_x$ is equivalent to a query to the oracle $U$ in \cref{lem:or}, with $\ket{y_1}$ taking the place of the control $\ket{c}$. Therefore, $\Qoip_0(f)$ queries to $U$ suffice to learn $x$. Having learned $x$, we can output $\mathsf{OR}(x_{11}, \dots, x_{1q})$, which (by the definition of $x$) is exactly $\mathsf{OR}(h(1), \dots, h(q))$. Therefore, by \cref{lem:or}, $\Qoip_0(f) \geq q$.
		
		Next, we show that $\Qcc_0(f) \leq \lceil \frac{1}{2} n \rceil + \lceil \frac{1}{2} \log q \rceil$. By superdense coding \cite{bw92}, Bob sends $j$ to Alice using $\lceil \frac{1}{2} \log q \rceil$ qubits. By superdense coding again, Alice sends $x_{1j}, \dots, x_{nj}$ to Bob using $\lceil \frac{1}{2} n \rceil$ qubits. Finally, Bob outputs $\sum_{i = 1}^n y_i \cdot x_{ij} \pmod{2}$.
		
		Combining completes the proof:
		\begin{align*}
			\Qcc_0(f) \cdot \Qoip_0(f) &\leq \left(\frac{1}{2}n + \frac{1}{2} \log q + O(1)\right) \cdot q \\
			&\leq \frac{1}{2} nq + O(q \log q) \\
			&\leq \frac{1}{2} \log |X| + O(\Qoip_0(f) \cdot \log \Qoip_0(f)). \qedhere
		\end{align*}
	\end{proof}

	\subsection{The bounded-error case}
	A bounded-error version of \cref{thm:zero-error-optimality} is also true by a similar argument. To prove it, we use the following lemma by Farhi et al. and Beals et al. in place of \cref{lem:or}:
	\begin{lem}[{\cite{fggs98b, bbc+01}}] \label{lem:parity}
		Suppose $h: [q] \to \{0, 1\}$ is an unknown function. Let $U$ be the unitary operator on $\C^q \otimes \C^2 \otimes \C^2$ defined by
		\[
		U \ket{j} \ket{c} \ket{a} = \ket{j} \ket{c} \ket{a \oplus c \cdot h(j)}.
		\]
		Then at least $q/2$ queries to $U$ are required to compute $\mathsf{PARITY}(h(1), \dots, h(q))$ with failure probability at most $1/3$.
	\end{lem}

	\begin{theorem} \label{thm:nonzero-error-optimality}
		For any finite set $X$ and any positive integer $q \leq \log |X|$, there exists a finite set $Y$ and a function $f: X \times Y \to \{0, 1\}$ such that $q/2 \leq \Qoip(f) \leq q$ and
		\[
			\Qoip(f) \cdot \Qcc(f) \leq \frac{1}{2} \log |X| + O(\Qoip(f) \cdot \log \Qoip(f)).
		\]
	\end{theorem}

	\begin{proof}[Proof sketch]
		Use the same construction as in the proof of \cref{thm:zero-error-optimality}. Use \cref{lem:parity} instead of \cref{lem:or} for the lower bound on $\Qoip(f)$. The rest of the analysis goes through as before, because obviously $\Qoip(f) \leq \Qoip_0(f)$ and $\Qcc(f) \leq \Qcc_0(f)$.
	\end{proof}
	
	\subsection{Open problems}
	We leave open the problem of closing the gaps between \cref{thm:zero-error,thm:zero-error-optimality} and between \cref{thm:main,thm:nonzero-error-optimality}. In particular:
	\begin{itemize}
		\item Is the $\log \Qoip(f)$ term in \cref{thm:main} necessary? Is the $\log |Z|$ term necessary?
		\item Can \cref{thm:main} be \emph{significantly} strengthened in the regime $\Qcc(f) \leq o(\log \log |X|)$? For example, if $f: \{0, 1\}^n \times \{0, 1\}^n \to \{0, 1\}$ satisfies $\Qcc(f) \leq O(1)$, does that imply that $\Qoip(f) \geq \exp(\Omega(n))$?
	\end{itemize}
	We also leave it as an open question whether \cref{thm:distributional-tradeoff} (the distributional generalization of \cref{thm:main}) is near-optimal.
	
	\section{Acknowledgments}
	
	We thank Ronald de Wolf, Scott Aaronson, Alex Arkhipov, Jalex Stark, and Robin Kothari for helpful comments on a draft of this paper and for helpful discussions. We thank Thomas Vidick for pointing out that the entropy measure that appears in this paper has been studied before under the name ``smooth max-entropy''.

	\bibliographystyle{alpha}
	\bibliography{quantum-tradeoff}
\end{document}